\documentclass[11pt]{article}
\usepackage{a4wide}

\usepackage{amssymb,amsmath,amsthm,enumerate}

\usepackage{cleveref}

\usepackage{authblk}

\usepackage{ourmacros}

\date{}

\begin{document}

\title{Simplest Non-Regular Deterministic Context-Free Language}
\author[1]{Petr Jan\v{c}ar}
\author[2]{Ji\v{r}\'{\i} \v{S}\'{\i}ma}
\affil[1]{Dept of Computer Science, Faculty of Science, Palack\'y University Olomouc,
Czechia\par 
    \texttt{petr.jancar@upol.cz}}
\affil[2]{Institute of Computer Science of the Czech Academy of Sciences, Prague, 
Czechia\par
    \texttt{sima@cs.cas.cz}}

\renewcommand{\Affilfont}{\itshape\small}
\maketitle           

\bibliographystyle{splncs04}

\begin{abstract}
We introduce a new notion of $\mathcal{C}$-sim\-ple problems	for a class $\mathcal{C}$ of decision problems (i.e. languages), w.r.t. a particular
	reduction.
A problem is $\mathcal{C}$-simple if it can be reduced to each problem in $\mathcal{C}$. This can be viewed as a conceptual
	counterpart to $\mathcal{C}$-hard problems to which all
	problems in $\mathcal{C}$ reduce. Our concrete example is
	 the class of \emph{non-regular} deterministic
	context-free languages (\nrdcfl), with 
	a truth-table reduction by 
	Mealy machines
	(which proves to be a preorder).
 The main technical
result is a proof that the {\nrdcfl} language
	$L_\#=\{0^n1^n\mid n\geq 1\}$ is \nrdcfl-simple,
which can thus be viewed as the simplest problem in the class \nrdcfl.

	This result has already provided an application, to the computational
	model of neural networks 1ANN at the first level of analog
	neuron hierarchy. This model was proven not to recognize
	$L_\#$, by
	using a specialized technical argument that can hardly be
	generalized to other languages in \nrdcfl. 
By the result that $L_\#$ is \nrdcfl-simple, w.r.t. the reduction that
	can be implemented by 1ANN,
	we immediately obtain that
 1ANN cannot accept any
language in \nrdcfl.

	It thus seems worthwhile to explore if looking for
	$\mathcal{C}$-simple problems in other classes $\mathcal{C}$
	under suitable reductions could provide effective tools for expanding the lower-bound results known for single problems to the whole classes of problems.

\medskip

\end{abstract}

\section{Introduction}

We introduce a new notion of 
\emph{$\mathcal{C}$-simple problems}
for a class $\mathcal{C}$ of decision problems (i.e. languages). A problem is $\mathcal{C}$-simple 
if it can be reduced to each problem in~$\mathcal{C}$; if this problem is, moreover, in $\mathcal{C}$, it can be viewed as a simplest problem in $\mathcal{C}$.  
The $\mathcal{C}$-simple problems are thus a conceptual counterpart
to the common $\mathcal{C}$-hard problems (like, e.g., NP-hard
problems) to which conversely any problem in $\mathcal{C}$ reduces.
These definitions (of $\mathcal{C}$-simple and $\mathcal{C}$-hard
problems) are parametrized by a
chosen reduction
that does not have a higher computational complexity than the class
$\mathcal{C}$ itself.
Therefore, it may be said that if a $\mathcal{C}$-hard problem has a
(computationally) ``easy'' solution, then each problem in
$\mathcal{C}$ has an ``easy'' solution. On the other hand, if we prove
that a $\mathcal{C}$-simple problem is not ``easy'', in particular
that it
cannot be solved by machines of a type $\calM$ that can implement the
respective reduction, then all problems in $\mathcal{C}$ are not 
``easy'', that is, are not
solvable by $\calM$; this extends a lower-bound result for one problem to the whole class of problems.

In this paper, we consider $\mathcal{C}$ to be the class of
non-regular deterministic context-free languages, which we denote by
\nrdcfl; we thus have \nrdcfl\,=\,DCFL\,$\smallsetminus$\,REG
(where REG denotes the class of regular languages).
We use
a~truth-table reduction by Mealy machines (which is motivated below).
Hence  a~{\nrdcfl}-simple problem is a language $L_0\subseteq\Sigma^*$ (over
an alphabet $\Sigma$) that can be reduced to each {\nrdcfl}
language $L\subseteq\Delta^*$ by a Mealy machine $\mathcal{A}$
 with an oracle $L$, denoted $\mathcal{A}^L$.
More precisely, the finite-state transducer $\mathcal{A}$ transforms a
given input word $w\in\Sigma^*$ to a prefix
$\mathcal{A}(w)\in\Delta^*$ of queries for the oracle $L$. 
In addition, each state $q$ of $\mathcal{A}^L$ is associated with a
finite tuple $\sigma_q=(s_{q1},\ldots,s_{qr_q})$ of $r_q$ query
suffixes from $\Delta^*$, and with a truth table $f_q:\{0,1\}^{r_q}\rightarrow\{0,1\}$. 
After $\mathcal{A}^L$ reads an input word~$w$ (translating it to
$\mathcal{A}(w)$), by which it enters a state $q$,
for each $i\in\{1,2,\dots,r_q\}$ it queries whether or not the string
$\mathcal{A}(w)\cdot s_{qi}$ is in $L$
(or, equivalently, whether or not $\mathcal{A}(w)$ belongs to the
quotient $L/s_{qi}=\{v\in\Delta^*\mid v\cdot s_{qi}\in L\}$), 
 and aggregates the answers by the truth table $f_q$ for deciding if
 $w$ is accepted.

This truth-table reduction by Mealy machines proves to be a preorder,
denoted as $\mpred$.
The main technical result of this paper is that the 
{\nrdcfl} language $L_\#=\{0^n1^n\mid n\geq 1\}$ 
(over the binary alphabet $\{0,1\}$) 
is \nrdcfl-simple, since $L_\#\mpred L$ for each 
language $L$ in \nrdcfl. 
The class DCFLS of \nrdcfl-simple languages
comprises REG 
and is
a strict subclass of DCFL; e.g., the {\nrdcfl} language
$L_R=\left\{wcw^R\mid w\in\{a,b\}^*\right\}$ over the alphabet
$\{a,b,c\}$ proves to be not \nrdcfl-simple. The closure properties of
DCFLS are similar to that of DCFL as the class DCFLS is closed under
complement and intersection with regular languages, while being not closed under concatenation, intersection, and union.

The above definition of \nrdcfl-simple problems has originally been
motivated by the analysis of the computational power of neural network
(NN) models which is known to depend on the (descriptive) complexity
of their weight parameters~\cite{Siegelmann99,SimaO03}. The so-called
analog neuron hierarchy~\cite{Sima20} of binary-state NNs with
increasing number of $\alpha$ extra analog-state neurons, denoted as
$\alpha$ANN for $\alpha\geq 0$, has been introduced for studying NNs
with realistic weights between integers (finite automata) and rational
numbers (Turing machines). We use the notation $\alpha$ANN also for the class of languages accepted by $\alpha$ANNs, which can clearly be distinguished by the context. The separation 1ANN~$\subsetneq$~2ANN has
been witnessed by the {\nrdcfl} language $L_\#\in$~2ANN~$\setminus$~1ANN. The proof of $L_\#\notin$~1ANN is rather technical (based on the Bolzano-Weierstrass theorem) which could hardly be generalized to other {\nrdcfl} languages, while it was conjectured that $L\notin$~1ANN for all {\nrdcfl} languages $L$, that is,
$\mbox{\nrdcfl}\subseteq(\mbox{2ANN}\,\setminus\,\mbox{1ANN})$
(implying $\mbox{1ANN}\,\cap\,\mbox{DCFL}=\mbox{0ANN}=\mbox{REG}$). 
An idea how to prove this conjecture is to show that $L_\#\notin$~1ANN
is in some sense the simplest problem in the class
$\mbox{\nrdcfl}$, namely, to reduce $L_\#$ to any
{\nrdcfl} language $L$ by using a reduction that can be carried
out by 1ANNs, 
which are at least as powerful as finite automata.
This would imply that $L$ cannot be accepted by any 1ANN since it is
at least as hard as $L_\#$ that has been proven not to be recognized by 1ANNs. 

The idea why $L_\#$ should serve as the simplest language in the class
{\nrdcfl} comes from the fact that any reduced context-free grammar $G$ generating a non-regular language
$L\subseteq\Delta^*$ is
\emph{self-embedding}~\cite[Theorem~4.10]{HoUl69}. This means that
there is a so-called self-embedding nonterminal $A$ admitting the
derivation $A\Rightarrow^*xAy$ for some non-empty strings
$x,y\in\Delta^+$. Since $G$ is reduced, there are strings
$v,w,z\in\Delta^*$ such that $S\Rightarrow^*vAz$ and $A\Rightarrow^*w$
where $S$ is the start nonterminal in $G$, which implies
$S\Rightarrow^*vx^mwy^mz\in L$ for every $m\geq 0$. It is thus
straightforward to suggest to reduce an input word
$0^m1^n\in\{0,1\}^*$ where $m,n\geq 1$, to the string
$vx^mwy^nz\in\Delta^*$ (while the inputs outside
$0^+1^+$ are mapped onto some fixed string
outside $L$) since $0^m1^n\in L_\#$ entails $vx^mwy^nz\in L$.

However, the suggested (one-one) reduction from $L_\#$ to $L$ is not consistent
because $vx^mwy^nz\in L$ does not necessarily imply $0^m1^n\in L_\#$.
For example, consider the {\nrdcfl} language $L_1=\{0^m1^n\mid
1\leq m\leq n\}$ over the binary alphabet $\Delta=\{0,1\}$ for which
there are no words $v,x,w,y,z\in\Delta^*$ such that $vx^mwy^nz\in L_1$
would ensure $m=n$. Nevertheless, we can pick two inputs $0^m1^{n-1}$
and $0^m1^n$ instead of one, that is, $x=0$, $y=1$, and
$v=w=z=\varepsilon$ ($\varepsilon$ denoting the empty string), which satisfy $0^m1^n\in L_\#$ if{f} $m=n$ if{f} $vx^mwy^{n-1}z\notin L_1$ and $vx^mwy^nz\in L_1$. 
It turns out that
this can be generalized to any {\nrdcfl}
language. Namely, we prove in this paper that for {\nrdcfl} language $L\subseteq\Delta^*$ over any alphabet~$\Delta$, there
are non-empty words $v,x,w,y,z\in\Delta^+$ and a language
$L'\in\{L,\overline{L}\}$, where 
$\overline{L}=\Delta^*\smallsetminus L$ is the complement of $L$,
such that $0^m1^n\in L_\#$ if{f} $vx^mwy^{n-1}z\notin L'$ and $vx^mwy^nz\in L'$. 

Therefore, the simple many-one (in fact, one-one) 
reduction from $L_\#$ with one query to the oracle $L$ is replaced
by a truth-table reduction, that is, by 
a special Turing
reduction in which all its finitely many (in our case two) oracle
queries are presented at the same time and there is a Boolean function
(a truth table) which, when given the answers to the queries, produces the final answer of the reduction. This truth-table reduction from $L_\#$ to $L$ can be implemented by a  deterministic finite-state transducer 
(a Mealy machine) $\mathcal{A}$ with the oracle $L$:
It transforms the input $0^m1^n$ where $m,n\geq 1$ (the inputs outside $0^+1^+$ are rejected), to the output $vx^mwy^{n-1}\in\Delta^+$
and carries out two queries to $L$ that arise by concatenation of this
output with two fixed suffixes $z$ and $yz$; hence the queries are
$vx^mwy^{n-1}z\stackrel{?}{\in}L$ and $vx^mwy^nz\stackrel{?}{\in}L$.
The truth table is defined so that the input $0^m1^n$ is accepted by
$\mathcal{A}^L$ if{f} the two answers to these queries are distinct
and at same time, the first answer is negative in the case $L'=L$,
and positive in the case $L'=\overline{L}$, which is equivalent to
$0^m1^n\in L_\#$.

It follows that the {\nrdcfl} language $L_\#$ is {\nrdcfl}-simple
 under the truth-table reduction by Mealy machines. Since this
reduction can be implemented by 1ANNs, we achieve the desired stronger
separation
$\mbox{\nrdcfl}\subseteq(\mbox{2ANN}\,\setminus\,\mbox{1ANN})$
in the analog neuron hierarchy~\cite{Sima21}. This result constitutes
a non-trivial application of the proposed concept of \nrdcfl-simple
problem. Moreover, if we could generalize the result to
(nondeterministic) context-free languages (CFL), e.g. by proving that
some {\nrdcfl} language is \nrcfl-simple (where
{\nrcfl}\,$=$\,CFL$\smallsetminus$\,REG), which would imply that
$L_\#$ is \nrcfl-simple by the transitivity of reduction, then we would
achieve even stronger separation
$\mbox{\nrcfl}\subseteq(\mbox{2ANN}\,\setminus\,\mbox{1ANN})$. 
We note the interesting fact that $L_\#$ cannot be CSL$'$-simple (under
our reduction), since
1ANN accepts some context-sensitive
languages outside CFL~\cite{Sima20}.

In general, if we show that some $\mathcal{C}$-simple problem under a
given reduction cannot be computed by a computational model $\calM$
that implements this reduction, then all 
problems in
the class $\mathcal{C}$ are not solvable by $\calM$ either. The notion
of $\mathcal{C}$-simple problems can thus be useful for expanding known (e.g.\ technical) lower-bound results for individual problems to the whole classes of problems
at once, as it was the case of the \nrdcfl-simple problem
$L_\#\notin\,\mbox{1ANN}$, expanding to $\mbox{\nrdcfl}\cap\,\mbox{1ANN}\,=\emptyset$. 
It seems worthwhile to explore if looking for
$\mathcal{C}$-simple problems in other complexity classes
$\mathcal{C}$ could provide effective tools for strengthening known
lower bounds.

We remark that the hardest
context-free language by Greibach~\cite{DBLP:journals/siamcomp/Greibach73} can be viewed as CFL-hard 
under a special type of our reduction $\mpred$.
Related line of study concerns the types of reductions used in finite or pushdown automata with oracle. For example, nondeterministic finite automata with oracle complying with many-one restriction have been applied to establishing oracle hierarchies over the context-free languages~\cite{Reinhardt90}. For the same purpose, oracle pushdown automata have been used for many-one, truth-table, and Turing reducibilities, respectively, inducing the underlying definitions also to oracle nondeterministic finite automata~\cite{Yamakami14}. In addition, nondeterministic finite automata whose oracle queries are completed by the prefix of an input word that has been read so far and the remaining suffix, have been employed in defining a polynomial-size oracle hierarchy~\cite{AnabtawiHKZ19}.

In the preliminary study~\cite{SimaP19}, some considerations about the
simplest {\nrdcfl} language have appeared, yet
without formal definitions of \nrdcfl-simple problems, that included
only
sketches of incomplete proofs of weaker results based on the
representation of DCFL by so-called deterministic monotonic restarting
automata~\cite{JancarMPV99}, which have initiated investigations of
non-regularity degrees in DCFL~\cite{MrazPPS20}.

In this paper we achieve a complete argument for $L_\#$ to be 
a \nrdcfl-simple problem, within the framework of deterministic pushdown
automata (DPDA) by using some ideas on regularity of pushdown
processes from~\cite{Jancar20}. We now give an informal overview of the
proof. Given a DPDA $\calM$ recognizing a non-regular language
$L\subseteq \Delta^*$, it
is easy to realize that some computations of $\calM$ (from the initial
configuration) must be reaching configurations where the stack
is arbitrarily large while it can be (almost) erased afterwards.
Hence the existence of words $v,x,w,y,z\in\Delta^+$ such that
$vx^mwy^mz\in L$ for all $m\geq 0$ is obvious. However, we aim 
to guarantee that for all $m,n$ the equality $m=n$ holds if, and only
if,
$vx^mwy^{n-1}z\notin L'$ and 
 $vx^mwy^{n}z\in L'$, where $L'$ is either the language $L$ or its
 complement.
 This is not so straightforward but it is confirmed by our 
 detailed analysis (in  \cref{mproof}). We study the computation of
 $\calM$ on an infinite word $a_1a_2a_3\cdots$ that visits infinitely
 many pairwise non-equivalent configurations. 
We use a natural congruence property of language equivalence on
the set of configurations, and avoid
some tedious technical
 details by a particular use of Ramsey's theorem. This allows us to
 extract the required tuple  $v,x,w,y,z\in\Delta^+$ from the mentioned
 infinite computation. We note that determinism of $\calM$ is
essential in the presented proof; we leave open if it can be
 relaxed to show that  $L_\#$ is even CFL$'$-simple.

The rest of the paper is organized as follows. 
In \cref{oMM} we recall basic definitions and notation regarding DPDA
and Mealy machines, introduce the novel concept of \nrdcfl-simple
problems under truth-table reduction by Mealy machines and show some
simple properties of the class DCFLS of \nrdcfl-simple problems. In
\cref{mproof} we present the proof of the main technical result which
shows that $L_\#$ is \nrdcfl-simple. Finally, we summarize the results and list some open problems in \cref{concl}.

\section{\nrdcfl-Simple 
Problem Under Truth-Table Mealy Reduction}
\label{oMM}

In this section we define the truth-table reduction by Mealy
machines, introduce the notion of \nrdcfl-simple 
problems, show their
basic properties, and formulate the main technical result (\cref{maint}).
But first we recall standard definitions of pushdown automata.

A \emph{pushdown automaton (PDA)} is a tuple
$\calM=(Q,\Sigma,\Gamma,R,q_0,X_0,F)$ where $Q$ is a finite set of
states including the start state $q_0\in Q$ and the set $F\subseteq Q$
of accepting states, while the finite sets $\Sigma\not=\emptyset$ and
$\Gamma\not=\emptyset$ represent the input and stack alphabets,
respectively, with the initial stack symbol $X_0\in\Gamma$. In
addition, the set $R$~contains finitely many transition rules
$pX\gt{a}q\gamma$ with the meaning that $\calM$ in state $p\in Q$, on
the input $a\in\Sigma_\varepsilon=\Sigma\cup\{\varepsilon\}$
(recall $\varepsilon$ denotes the empty string),
and with $X\in\Gamma$ as the topmost stack symbol may read $a$, change the state to $q\in Q$, and pop $X$, \mbox{replacing it by pushing $\gamma\in\Gamma^*$.}

By a \emph{configuration} of $\calM$ we mean $p\alpha\in
Q\times\Gamma^*$, 
and we define relations $\gt{a}$ for $a\in\Sigma_\varepsilon$ on 
$Q\times\Gamma^*$:
each rule $pX\gt{a}q\gamma$ in $R$ induces
$pX\alpha\gt{a}q\gamma\alpha$ for all $\alpha\in\Gamma^*$; these
relations are naturally extended to $\gt{w}$ for $w\in\Sigma^*$.
For a~configuration $p\alpha$ we define 
$\calL(p\alpha)=\{w\in\Sigma^*\mid p\alpha\gt{w}q\beta\mbox{ for some
}q\in F\mbox{ and }\beta\in\Gamma^*\}$,
and $\calL(\calM)=\calL(q_0X_0)$ is the language accepted by $\calM$.
A PDA $\calM$ is \emph{deterministic} (a DPDA) if there is at most one
rule $pX\gt{a}..$ for each tuple $p\in Q$, $X\in\Gamma$,
$a\in\Sigma_{\varepsilon}$; moreover, if there is a rule
 $pX\gt{\varepsilon}..$, then there is no rule $pX\gt{a}..$ for
 $a\in\Sigma$. 
We also use the standard assumption that all $\varepsilon$-steps are
popping, that is, in each rule $pX\gt{\varepsilon}q\gamma$ in $R$ we have $\gamma=\varepsilon$.

The languages accepted by (deterministic) pushdown automata constitute
the class of \emph{(deterministic) context-free languages}; the
classes are denoted by DCFL and CFL, respectively, whereas 
\nrdcfl\,$=$\,DCFL\,$\smallsetminus$\,REG. 

In the following theorem
we formulate the main technical result: any language in
{\nrdcfl} includes a certain ``projection'' of the language
$L_\#=\{0^n1^n\mid n\geq 1\}$, which means that $L_\#$ is in some
sense the simplest language in the class \nrdcfl. 
The theorem, whose proof will be presented in \cref{mproof}, 
thus provides an interesting property 
of \nrdcfl.
\begin{theorem}
\label{maint}
Let $L\subseteq\Delta^*$ be a non-regular deterministic context-free
	language over an alphabet~$\Delta$.
	There exist non-empty words $v,x,w,y,z\in\Delta^+$ and a language $L'\in\{L,\overline{L}\}$ 
	(where  $\overline{L}=\Delta^*\smallsetminus L$ is the complement of $L$)
 such that for all $m\geq 0$ and $n>0$ we have
\begin{equation}
\label{cond}
\left(vx^mwy^{n-1}z\notin L' \mbox{ and }\, vx^mwy^nz\in L'\right)\quad\mbox{if{f}}\quad m=n\,.
\end{equation}
\end{theorem}

In order to formalize the \nrdcfl-simple problems, we now define a
\emph{Mealy machine} $\mathcal{A}$ \emph{with an oracle}:  
it is a tuple
$\mathcal{A}=(Q,\Sigma,\Delta,\delta,\lambda,q_0,\{(\sigma_q,f_q)\mid
q\in Q\})$ where $Q$ is a finite set of states including the start
state $q_0\in Q$, and the finite sets $\Sigma\not=\emptyset$ and
$\Delta\not=\emptyset$ represent the input and output (oracle)
alphabets, respectively. Moreover, $\delta:Q\times\Sigma\rightarrow Q$
is a (partial) state-transition function which extends to input
strings as $\delta:Q\times\Sigma^*\rightarrow Q$ where
$\delta(q,\varepsilon)=q$ for every $q\in Q$, while
$\delta(q,wa)=\delta(\delta(q,w),a)$ for all $q\in Q$, $w\in\Sigma^*$, $a\in\Sigma$. Similarly, $\lambda:Q\times\Sigma\rightarrow\Delta^*$ is an output function which extends to input strings as $\lambda:Q\times\Sigma^*\rightarrow\Delta^*$ where $\lambda(q,\varepsilon)=\varepsilon$ 
for all $q\in Q$, and
$\lambda(q,wa)=\lambda(q,w)\cdot\lambda(\delta(q,w),a)$ for all $q\in Q$, $w\in\Sigma^*$, $a\in\Sigma$.
In addition, for each $q\in Q$, the tuple
$\sigma_q=(s_{q1},\ldots,s_{qr_q})$ of strings in $\Delta^*$ contains $r_q$~query suffixes, while $f_q:\{0,1\}^{r_q}\rightarrow\{0,1\}$ is a truth table that aggregates the answers to the $r_q$~oracle queries.

The above Mealy machine $\mathcal{A}$ starts in the start state $q_0$ and
operates as a deterministic finite-state transducer that transforms an
input word $w\in\Sigma^*$ to the output string
$\mathcal{A}(w)=\lambda(q_0,w)\in\Delta^*$ written to a so-called
oracle tape. The oracle tape is a semi-infinite, write-only tape which
is empty at the beginning and its contents are only extended in the
course of computation by appending the strings to the right. Namely,
given a current state $q\in Q$ and an input symbol $a\in\Sigma$, the
machine $\mathcal{A}$ moves to the next state $\delta(q,a)\in Q$ and
writes the string $\lambda(q,a)\in\Delta^*$ to the oracle tape, if
$\delta(q,a)$ is defined; otherwise $\mathcal{A}$ rejects the input.
After reading the whole input word $w\in\Sigma^*$, the
machine~$\mathcal{A}$ is in the state $p=\delta(q_0,w)\in Q$, while the oracle tape contains the output $\mathcal{A}(w)=\lambda(q_0,w)\in\Delta^*$. 

Finally, the Mealy machine $\mathcal{A}$, equipped with an
oracle $L\subseteq \Delta^*$, in this case denoted $\mathcal{A}^L$, 
queries the oracle whether $\mathcal{A}(w)$ belongs to the (right)
quotient $L/s_{pi}=\{u\in\Delta^*\mid u\cdot s_{pi}\in L\}$, for each
suffix $s_{pi}$ in $\sigma_p$, and the answers are aggregated by the
truth table $f_p$. Thus, the oracle Mealy machine $\mathcal{A}^L$ accepts the input
word $w\in\Sigma^*$ if{f}
\[
f_p\left(\chi_{L/s_{p1}}(\mathcal{A}(w)),\chi_{L/s_{p2}}(\mathcal{A}(w)),\ldots,\chi_{L/s_{pr_p}}(\mathcal{A}(w))\right)=1
\]
where $p=\delta(q_0,w)$ and
$\chi_{L/s_{pi}}:\Delta^*\rightarrow\{0,1\}$ is the characteristic function of $L/s_{pi}$, 
that is, $\chi_{L/s_{pi}}(u)=1$ if $u\cdot s_{pi}\in L$, and
$\chi_{L/s_{pi}}(u)=0$ if $u\cdot s_{pi}\notin L$.
The language accepted by the machine $\mathcal{A}^L$ 
is defined as 
$\calL(\mathcal{A}^L)=\{w\in\Sigma^*\mid w
$ is accepted by $\mathcal{A}^L\}$.\footnote{Note that the described protocol works also for non-prefix-free languages since for any input prefix that has been read so far, the output value from the truth table determines whether the oracle Mealy machine is in an ``accepting'' state, deciding about this prefix analogously as a deterministic finite automaton. The truth-table reduction only requires that the given oracle answers do not influence further computation when subsequent input symbols are read.}

We say that 
$L_1\subseteq\Sigma^*$ is \emph{truth-table reducible} to $L_2\subseteq\Delta^*$ 
\emph{by a Mealy machine}, which is denoted  as $L_1\mpred L_2$,
if $L_1=\calL(\mathcal{A}^{L_2})$ for some Mealy machine~$\mathcal{A}$
running with the oracle $L_2$. The following lemma shows that we can
chain these reductions together since the relation $\mpred$ is a preorder.
\begin{lemma}
\label{trans}
The relation $\mpred$ is reflexive and transitive. 
\end{lemma}

\begin{proof}
The relation $\mpred$ is reflexive since
	$L=\calL(\mathcal{A}^L)\subseteq\Sigma^*$ for the oracle Mealy
	machine $\mathcal{A}^L=(\{q\},\Sigma,\Sigma,\delta,\lambda,q,\{(\sigma_q,f_q)\})$ where $\delta(q,a)=q$ and $\lambda(q,a)=a$ for every $a\in\Sigma$, $\sigma_q=(\varepsilon)$, and $f_q$ is the identity.

Now we show that the relation $\mpred$ is transitive. 
	Let $L_1\mpred L_2$ and $L_2\mpred L_3$ which means $L_1=\calL(\mathcal{A}_1^{L_2})\subseteq\Sigma^*$ and $L_2=\calL(\mathcal{A}_2^{L_3})\subseteq\Delta^*$ for some oracle Mealy machines $\mathcal{A}_1^{L_2}=(Q_1,\Sigma,\Delta,\delta_1,\lambda_1,q_0^1,\{(\pi_q,g_q)\mid q\in Q_1\})$ and $\mathcal{A}_2^{L_3}=(Q_2,\Delta,\Theta,\delta_2,\lambda_2,q_0^2,\{(\varrho_q,h_q)\mid q\in Q_2\})$, respectively. We will construct the oracle Mealy machine~$\mathcal{A}^{L_3}=(Q,\Sigma,\Theta,\delta,\lambda,q_0,\{(\sigma_q,f_q)\mid q\in Q\})$ such that $L_1=\calL(\mathcal{A}^{L_3})\subseteq\Sigma^*$ which implies the transitivity $L_1\leq^\mathcal{A} L_3$.
We define $Q=Q_1\times Q_2$ with $q_0=(q_0^1,q_0^2)$, $\delta((q_1,q_2),a)=(\delta_1(q_1,a),\delta_2(q_2,\lambda_1(q_1,a)))$
and $\lambda((q_1,q_2),a)=\lambda_2(q_2,\lambda_1(q_1,a))$ for every $(q_1,q_2)\in Q$ and $a\in\Sigma$, which ensures $\mathcal{A}(w)=\lambda(q_0,w)=\lambda_2(q_0^2,\lambda_1(q_0^1,w))=\mathcal{A}_2(\mathcal{A}_1(w))\in\Theta^*$ for every $w\in\Sigma^*$. For each state $p=(p_1,p_2)\in Q$ in $\mathcal{A}$, we define the tuple of query suffixes from $\Theta^*$,
\[
\sigma_p=\left(\lambda_2(p_2,s_{p_1,i})\cdot s_{p_2(i),j}\,\big|\,i=1,\ldots,r_{p_1}\,,\,j=1,\ldots,r_{p_2(i)}\right)
\]
where $\pi_{p_1}=(s_{p_1,1},s_{p_1,2}\ldots,s_{p_1,r_{p_1}})\in\Delta^{r_{p_1}}$ and $\varrho_{p_2(i)}=(s_{p_2(i),1},s_{p_2(i),2}\ldots,s_{p_2(i),r_{p_2(i)}})\in\Theta^{r_{p_2(i)}}$ are the query suffixes associated with $p_1\in Q_1$ and $p_2(i)=\delta_2(p_2,s_{p_1,i})\in Q_2$ for $i\in\{1,\ldots,r_{p_1}\}$, respectively, and the truth table $f_p=g_{p_1}(h_{p_2(1)},\ldots,h_{p_2(r_{p_1})})$ aggregates the answers to the corresponding oracle queries, which ensures $L_1=\calL(\mathcal{A}^{L_3})\subseteq\Sigma^*$.
\end{proof}

We say that a (decision) problem $L_0\subseteq\Sigma^*$ is
\emph{\nrdcfl-simple} if $L_0\mpred L$ for every non-regular
deterministic context-free language $L\subseteq\Delta^*$. It follows
from \cref{maint} that the {\nrdcfl} language $L_\#$ is an
example of a \nrdcfl-simple problem. In addition, we denote by DCFLS the class of \nrdcfl-simple problems and formulate its basic properties.
\begin{corollary}[of \cref{maint}]
\label{dcflsimpl}
The non-regular deterministic context-free language $L_\#=\{0^n1^n\mid n\geq 1\}$ is \nrdcfl-simple.
\end{corollary}

\begin{proof}
Let $L\subseteq\Delta^*$ be any {\nrdcfl} language. According
	to \cref{maint}, there are $v,x,w,y,z\in\Delta^+$ and
	$L'\in\{L,\overline{L}\}$ such that condition~(\ref{cond})
	holds for $L'$.
We define the Mealy machine $\mathcal{A}^L=(\{q_0,q_1,q_2\},\{0,1\},\Delta,\delta,\lambda,q_0,\{(\sigma_q,f_q)\mid q\in Q\})$ with the oracle $L$, as $\delta(q_0,0)=\delta(q_1,0)=q_1$, $\delta(q_1,1)=\delta(q_2,1)=q_2$, $\lambda(q_0,0)=vx$, $\lambda(q_1,0)=x$, $\lambda(q_1,1)=w$, $\lambda(q_2,1)=y$,
$\sigma_{q_2}=(z,yz)$,
$f_{q_0}=f_{q_1}=0$, $f_{q_2}(0,0)=f_{q_2}(1,1)=0$, 
and $f_{q_2}(1,0)=1-f_{q_2}(0,1)$ where $f_{q_2}(0,1)=1$ if{f} $L'=L$.
	It is easy to verify that $L_\#=\calL(\mathcal{A}^L)$, which
	implies $L_\#\mpred L$. Hence, $L_\#$ is \nrdcfl-simple.
\end{proof}

\begin{proposition}~
\label{dcfls}
\begin{enumerate}
\item $\mbox{REG}\,\subsetneq\,\mbox{DCFLS}$.
\item $\mbox{DCFLS}\,\subsetneq\,\mbox{DCFL}$, and $L_R=\{wcw^R\mid
	w\in\{a,b\}^*\}\in\mbox{DCFL}\smallsetminus\mbox{DCFLS}$.
\item \label{compl} The class DCFLS is closed under complement and intersection with regular languages.
\item The class DCFLS is not closed under concatenation, intersection and union.
\end{enumerate}
\end{proposition}
\begin{proof}[Sketch.]~\\
\textbf{1.}
For any regular language $L$, consider a Mealy machine $\mathcal{A}^{L_\#}$ with the \nrdcfl-simple oracle $L_\#$, that simulates a deterministic finite automaton recognizing $L$, while its constant truth tables produce~1 if{f} associated with the accept states. Hence, $L\mpred L_\#$ which means $L$ is \nrdcfl-simple according to \cref{trans} and \cref{dcflsimpl} which also implies $\mbox{REG}\,\not=\,\mbox{DCFLS}$.

\noindent
\textbf{2.}
We first observe that $\mbox{DCFLS}\,\subseteq\,\mbox{DCFL}$. Let $L\in$\,DCFLS be any \nrdcfl-simple language which ensures $L\mpred L_\#$ by an oracle Mealy machine $\mathcal{A}^{L_\#}$. The machine $\mathcal{A}^{L_\#}$ can be simulated by a DPDA $\calM$ which extends a suitable DPDA $\calM_\#$ (e.g.\ with no $\varepsilon$-transitions) accepting $L_\#=\calL(\calM_\#)$, so that the finite control of $\calM$ implements the finite-state transducer $\mathcal{A}$ whose output is presented online as an input to $\calM_\#$. Moreover, for each state $q$ of $\mathcal{A}$, the finite control of $\calM$ evaluates the truth table $f_q$ which aggregates the answers to the queries with $r_q$ suffixes associated with $q$, by inspecting at most constant number of topmost stack symbols. Hence $L=\calL(\calM)\in$\,DCFL.

In order to show that $\mbox{DCFLS}\,\not=\,\mbox{DCFL}$, we prove that the DCFL $L_R=\{wcw^R\mid w\in\{a,b\}^*\}$ over the alphabet $\{a,b,c\}^*$ is not \nrdcfl-simple. 
For the sake of contradiction, suppose that $L_R\mpred L_\#$ by a Mealy machine $\mathcal{A}^{L_\#}=(Q,\{a,b,c\}^*,\{0,1\}^*,\delta,\lambda,q_0,\{(\sigma_q,f_q)\mid q\in Q\})$ with the oracle $L_\#=\{0^n1^n\mid n\geq 1\}$, which means $L_R=\calL(\mathcal{A}^{L_\#})$. Consider all the $2^k$ possible prefixes $w\in\{a,b\}^k$ of inputs presented to $\mathcal{A}^{L_\#}$ that have the length $|w|=k$. These strings can bring $\mathcal{A}^{L_\#}$ into a finite number $|\{\delta(q_0,w)\mid w\in\{a,b\}^k\}|\leq|Q|$ of distinct states while the length $|\lambda(q_0,w)|$ of outputs written to the oracle tape is bounded by $O(k)$. For $\lambda(q_0,w)$ outside $0^*1^*$,  
the acceptance of words $wu$ where $u\in\{a,b,c\}^*$, depends only on the truth values $f_q(0,\ldots,0)$ associated with the states $q$ from the finite set~$Q$, due to $\lambda(q_0,wu)\notin L_\#/s$ for any $s\in\{0,1\}^*$. On the other hand, the number of distinct outputs $\lambda(q_0,w)$ in $0^*1^*$
is bounded by $O(k)$. This means that for a sufficiently large $k\geq 1$, there must be two distinct prefixes $w_1,w_2\in\{a,b\}^k$ such that 
$\delta(q_0,w_1)=\delta(q_0,w_2)$ and $\lambda(q_0,w_1)=\lambda(q_0,w_2)$ in $0^*1^*$,  
which results in the contradiction $w_1cw_2^R\in\calL(\mathcal{A}^{L_\#})\smallsetminus L_R$.

\noindent
\textbf{3.}
The class DCFLS is closed under complement since the truth tables can be negated. Furthermore, any oracle Mealy machine be can modified so that it simulates another given finite automaton in parallel and is forced to reject if this automaton rejects, which shows DCFLS to be closed under intersection with regular languages. 

\noindent
\textbf{4.}
	Observe that $(L_\#)^2$ is not \nrdcfl-simple under truth-table reduction. In addition, $L_1=\{0^m1^m0^n\mid m,n\geq 1\}$ and $L_2=\{0^m1^n0^n\mid m,n\geq 1\}$ are \nrdcfl-simple while $L_1\cap L_2$ is not context-free. The proof for union follows from \ref{compl} and De Morgan's law.
\end{proof}

\section{Proof of the Main Result (Theorem~\ref{maint})}
\label{mproof}

Theorem~\ref{maint} follows from the (more specific) next lemma 
that we prove in this section. 

By $\Nat$ we denote the set
$\{0,1,2,\dots\}$, and by $[i,j]$ the set $\{i,i{+}1,\dots,j\}$ (for
$i,j\in\Nat$).

\begin{lemma}\label{lem:nonregwitness}
Let $\calM=(Q,\Sigma,\Gamma,R,p_0,X_0,F)$ be a DPDA where
	$L=\calL(p_0X_0)$ is non-regular (hence $L$ belongs to
	\nrdcfl).
There are 
$v\in\Sigma^*$, $x,w,y,z\in\Sigma^+$, $p,q\in Q$,
$X\in\Gamma$,
$\gamma\in\Gamma^+$, $\delta\in\Gamma^*$ such that the following four
	conditions hold:
\begin{enumerate}
\item
$p_0X_0\gt{v}pX\delta$ and $pX\gt{x} pX\gamma$,
\\
which entails the infinite (stack increasing) computation
\begin{equation}\label{eq:witness}
p_0X_0\gt{v}pX\delta\gt{x}pX\gamma\delta
\gt{x}pX\gamma\gamma\delta\gt{x}pX\gamma\gamma\gamma\delta\gt{x}\cdots;
\end{equation}
\item
	$pX\gt{w}q$;

\item $q\gamma\gt{y}q$, 

hence	
	$q\gamma^\ell\delta'\gt{y^\ell}q\delta'$ for
		all $\ell\in\Nat$ and
		$\delta'\in\Gamma^*$;
	\item
		one of the following cases is valid (depending on
		whether
		$z\in\calL(q\delta)$ or $z\not\in\calL(q\delta)$):
		\begin{enumerate}
			\item
$\calL(q\gamma^k\delta)\ni y^\ell z$ iff $k=\ell$	
				(for all $k,\ell\in\Nat$),
				or 
$\calL(q\gamma^k\delta)\ni y^\ell z$ iff $k\leq\ell$	
				(for all $k,\ell\in\Nat$);

	\item
$\calL(q\gamma^k\delta)\ni y^\ell z$ iff $k\neq\ell$	
				(for all $k,\ell\in\Nat$),
				or 
$\calL(q\gamma^k\delta)\ni y^\ell z$ iff $k>\ell$	
				(for all $k,\ell\in\Nat$).
		\end{enumerate}
		\end{enumerate}		
\end{lemma}

\medskip

We note that
$p_0X_0\gt{v}pX\delta\gt{x^m}pX\gamma^m\delta\gt{w}q\gamma^m\delta\gt{y^m}q\delta$
(for each $m\in\Nat$);
hence $vx^mwy^mz\in L$ iff $z\in\calL(q\delta)$ (since $z$ is nonempty).
Theorem~\ref{maint} indeed follows from the lemma: 
there is $L'\in\{L,\overline{L}\}$ such that 
	either $vx^mwy^nz\in L'$ iff $m=n$ (for all $m,n\in\Nat$),
 or $vx^mwy^nz\in L'$ iff
	$m\leq n$ (for all $m,n\in\Nat$).
	(In \cref{maint} we also stated that $v$ is nonempty. If 
 $v=\varepsilon$ here, then we simply take $vx$ and $yz$ as the new $v,z$,
 respectively.)

\paragraph*{Proof of Lemma~\ref{lem:nonregwitness}}
In the rest of this section we provide a proof of  
\cref{lem:nonregwitness}, assuming a fixed
DPDA $\calM=(Q,\Sigma,\Gamma,R,p_0,X_0,F)$ where 
$L=\calL(p_0X_0)$ is non-regular. 
The proof structure is visible from 
the auxiliary claims that we state and prove on the way.

\medskip

\emph{Convention.} W.l.o.g. we assume that $\calM$ always 
reads the whole input $w\in\Sigma^*$ from $p_0X_0$.
This can be accomplished in the standard way, by adding 
a special bottom-of-stack symbol $\bot$ and a (non-accepting)
fail-state.
(Each empty-stack configuration $q\varepsilon$ becomes $q\bot$, and
each originally stuck computation enters the fail-state where it
loops. We also recall that all $\varepsilon$-steps are 
popping, and thus infinite $\varepsilon$-sequences are impossible.)
Hence for
any infinite word $a_1a_2a_3\cdots$ in $\Sigma^\omega$ there is
the unique infinite computation of $\calM$ starting in $p_0X_0$; it  
stepwise reads the whole infinite word $a_1a_2a_3\cdots$.

\medskip

The \emph{left quotient of} $L$ \emph{by} $u\in\Sigma^*$ is the set 
	$u\backslash L=\{v\in\Sigma^*\mid uv\in L\}$; 
	concatenation 
	has priority over $\backslash$, hence  $u_1u_2\backslash
	L=(u_1u_2)\backslash L$. (The next claim is valid for any
	non-regular $L$.)

\begin{claim}\label{cl:infseq}
We can fix an infinite word $a_1a_2a_3\cdots$ in
	$\Sigma^\omega$ ($a_i\in\Sigma$) such that $a_1a_2\cdots a_i\backslash L\neq
	a_1a_2\cdots a_j\backslash L$ for all $i\neq j$.
\end{claim}
\begin{proof}
Let us consider the labelled transition system	
	$\calT=(\LQ(L),\Sigma,(\gt{a})_{a\in\Sigma})$ 
	where $\LQ(L)=\{u\backslash L\mid u\in\Sigma^*\}$
	and $\mathop{\gt{a}} = \{(L',a\backslash L')\mid L'\in\LQ(L)\}$. 
	(We recall that $L'=u\backslash L$ entails $a\backslash L'=
	ua\backslash L$.)
	Since $L$ is non-regular, the set of states 
	reachable from $L=\varepsilon\backslash L$ in $\calT$ is
	infinite. The out-degree of states in $\calT$ is finite
	(in fact, bounded by $|\Sigma|$), hence an application of
	K\"onig's lemma yields
	an infinite \emph{acyclic} path
	$L\gt{a_1}L_1\gt{a_2}L_2\gt{a_3}\cdots$.
\end{proof}

We call
a~\emph{configuration} $p\alpha$ of $\calM$ \emph{unstable} if 
$\alpha=Y\beta$ and $R$ contains a rule
$pY\gt{\varepsilon}q$ (we recall that $\varepsilon$-steps are only
popping); otherwise $p\alpha$ is \emph{stable}. 
Since $\calM$ is a \emph{deterministic} PDA,
for each unstable $p\alpha$ we can soundly define the \emph{stable
successor of} $p\alpha$ as 
the unique stable configuration $p'\alpha'$ where 
$p\alpha\gt{\varepsilon}p'\alpha'$ ($\alpha'$ being a suffix of
$\alpha$). The path $p\alpha\gt{\varepsilon}p'\alpha'$ might (not) go
via an accepting state (in $F$), hence
$\calL(p\alpha)=\calL(p'\alpha')$ or 
$\calL(p\alpha)=\{\varepsilon\}\cup\calL(p'\alpha')$.
(We note that the
configurations in the computation~(\ref{eq:witness}) that start with
$pX$ are necessarily stable.)

\begin{claim}\label{cl:atmosttwice}	
Each configuration is visited at most twice by 
	\begin{equation}\label{eq:infcomput}
	\textnormal{the 
	computation of $\calM$ from $p_0X_0$
		on $a_1a_2a_3\cdots$ that is fixed by \cref{cl:infseq}.} 
	\end{equation}
\end{claim}
\begin{proof}
	The computation~(\ref{eq:infcomput})  is infinite, 
	stepwise reading the whole word $a_1a_2a_3\cdots$, 
and it can be presented as
	\begin{center}	
$r_0\gamma_0\gt{a_1}r_1\gamma_1\gt{a_2}r_2\gamma_2\gt{a_3}\cdots$
(for $r_0\gamma_0=p_0X_0$)
	\end{center}
		where each $r_i\gamma_i$ is stable; each
	segment
$r_i\gamma_i\gt{a_{i+1}}r_{i+1}\gamma_{i+1}$ starts with a (visible)
$a_{i+1}$-step that is followed by a (maybe empty) sequence of
(popping) $\varepsilon$-steps via unstable configurations. Since such
an $\varepsilon$-sequence might go through an accepting state,
we can have $r_i\gamma_i=r_j\gamma_j$ for $i\neq j$ though 
$a_1a_2\cdots a_i\backslash L\neq
	a_1a_2\cdots a_j\backslash L$; in this case $L$ contains precisely one of
	the words $a_1a_2\cdots a_i$ and $a_1a_2\cdots a_j$, and the
	languages
	$a_1a_2\cdots a_i\backslash L$ and
$a_1a_2\cdots a_j\backslash L$ differ just on $\varepsilon$.
Nevertheless, this reasoning entails that we cannot have 
 $r_i\gamma_i=r_j\gamma_j=r_\ell\gamma_\ell$ for pairwise different
	$i,j,\ell$.

Since each segment $r_i\gamma_i\gt{a_{i+1}}r_{i+1}\gamma_{i+1}$
	visits any unstable configuration at most once and 
$r_{i+1}\gamma_{i+1}$ is the stable successor for all unstable
	configurations in the segment, we deduce that also each unstable
	configuration can be visited at most twice
	in the computation~(\ref{eq:infcomput}).
\end{proof}

\begin{claim}\label{cl:stairfactor}
	The computation~(\ref{eq:infcomput}) on
	$a_1a_2a_3\cdots$
	 can be ``stair-factorized'',
	that is, written 
	\begin{equation}\label{eq:stairseq}
p_0X_0\gt{v_0}p_1X_1\alpha_1\gt{v_1}p_2X_2\alpha_2\alpha_1
\gt{v_2}p_3X_3\alpha_3\alpha_2\alpha_1\gt{v_3}\cdots
	\end{equation}
	so that for each $i\in\Nat$ we have 	$v_i\in\Sigma^+$ 
and
	$p_iX_i\gt{v_i}p_{i+1}X_{i+1}\alpha_{i+1}$ 
where $\alpha_{i+1}$ is a nonempty suffix of the right-hand
		side of a rule in $R$ (i.e., a nonempty suffix of $\gamma$
		in a rule $pX\gt{a}q\gamma$).
\end{claim}
\begin{proof}
We consider the
	computation~(\ref{eq:infcomput}), and 
	call a stable configuration $pX\beta$ a \emph{level}, with
	\emph{position} $i\in\Nat$,
	if $p_0X_0\gt{a_1\cdots a_{i}}pX\beta$ and all configurations visited by the computation
	$pX\beta\gt{a_{i+1}a_{i+2}\cdots}$ after $pX\beta$ have the
	stack longer than $|X\beta|$; we note that each level $pX\beta$
	has a unique position $\textsc{pos}(pX\beta)$. 
Since  
each configuration is visited at most twice in~(\ref{eq:infcomput}),
	the set of levels is infinite, with elements 
$p'_0X'_0$, $p_1X_1\beta_1$, $p_2X_2\beta_2$, $\dots$
where
	$0\leq \textsc{pos}(p'_0X'_0)<\textsc{pos}(p_1X_1\beta_1)<\textsc{pos}(p_2X_2\beta_2)<\cdots$. 
	The computation~(\ref{eq:infcomput}) can thus be presented as 
	\begin{center}
		$p_0X_0\gt{v'_{0}}p'_0X'_0\gt{v''_{0}}p_1X_1\beta_1\gt{v_1}p_2X_2\beta_2
\gt{v_2}p_3X_3\beta_3\gt{v_3}\cdots$
	\end{center}
	where $|v'_0|=\textsc{pos}(p'_0X'_0)$,
	and $|v_0v_1\cdots v_{j-1}|=\textsc{pos}(p_jX_j\beta_j)$
	for $j\geq 1$, putting $v_0=v'_0v''_0$.

Each segment 
	$pX\beta\gt{v}p'X'\beta'$ between two neighbouring levels can
	be obviously written as
	$pX\beta\gt{a}q\gamma_1\gamma_2\beta\gt{v'}p'X'\gamma_2\beta$
	where $pX\gt{a}q\gamma_1\gamma_2$ is a rule in $R$, both 
	$\gamma_1$ and $\gamma_2$ are nonempty, $v=av'$, and 
	$q\gamma_1\gt{v'}p'X'$.
Hence the validity of the claim is clear.	
\end{proof}

We define the natural \emph{equivalence relation} $\sim$ \emph{on the set of
	configurations of $\calM$}: we put
$p\alpha\sim q\beta$ if $\calL(p\alpha)= \calL(q\beta)$.

We fix the presentation~(\ref{eq:stairseq}), 
	calling $p_iX_i\alpha_i\alpha_{i-1}\cdots\alpha_1$ the
	\emph{level-configurations} (for all $i\in\Nat$).
Since 
we have 
	$\calL(p_iX_i\alpha_i\alpha_{i-1}\cdots\alpha_1)\smallsetminus\{\varepsilon\}=
	(v_0v_1\cdots v_{i-1}\backslash L)\smallsetminus\{\varepsilon\}$,
	there cannot be three level-configurations in the same
	$\sim$-class (i.e., in the same equivalence class w.r.t. $\sim$).
	Hence any 
infinite set of level-configurations represents infinitely many
	$\sim$-classes.
Now we show a congruence-property that might enable to shorten
a level-configuration while keeping its $\sim$-class. 
We use the notation $\ds(p\alpha)$ (the ``down-states'' of
	$p\alpha$), 
	putting
\begin{center}	 
	$\ds(p\alpha)=\{q\mid p\alpha\gt{w} q$ for some
	$w\in\Sigma^*\}$.
\end{center}

\begin{claim}\label{cl:congruence}
If $q\gamma\sim q\gamma'$ for each $q\in\ds(p\beta)$, then 
$p\beta\gamma\sim p\beta\gamma'$.
\end{claim}
\begin{proof}
Let us consider $w\in\Sigma^*$. If $w\in\calL(p\beta)$, then 
 $w\in\calL(p\beta\mu)$ for all $\mu\in\Gamma^*$. 
If $w\not\in\calL(p\beta)$ and 
	there is no prefix $v$ of $w$ such
	that $p\beta\gt{v}q$, 
	then $w\not\in\calL(p\beta\mu)$ for all $\mu\in\Gamma^*$. 
If $w\not\in\calL(p\beta)$ and $w=vv'$ where 
	$pX\beta\gt{v}q$ (necessarily for some $q\in\ds(pX\beta)$),
then  $w\in\calL(p\beta\mu)$ iff  $v'\in
	\calL(q\mu)$. Hence the claim is clear.
\end{proof}	

The next claim is an immediate corollary.

\begin{claim}\label{cl:shorten}
	Any computation
$p_0X_0\gt{w_1}pX\beta_1\gt{w_2}pX\beta_2\beta_1\gt{w_3}p'X'\beta_3\beta_2\beta_1$
	where $pX\gt{w_2}pX\beta_2$ ($w_2\in\Sigma^+$), $pX\gt{w_3}p'X'\beta_3$, and
	$q\beta_2\beta_1\sim q\beta_1$ for each
	$q\in\ds(p'X'\beta_3)$ can be shortened to 
$p_0X_0\gt{w_1}pX\beta_1\gt{w_3}p'X'\beta_3\beta_1$ where 
$p'X'\beta_3\beta_1\sim p'X'\beta_3\beta_2\beta_1$. 
\end{claim}

The $i$-th level-configuration
	in~(\ref{eq:stairseq}) is 
reached by the computation
	$p_0X_0\gt{v_0v_1\cdots
	v_{i-1}}p_iX_i\alpha_i\alpha_{i-1}\cdots\alpha_1$. 
	It can happen that there are $j_1,j_2$, $0\leq j_1<j_2\leq i$ such
	that $p_{j_1}X_{j_1}=p_{j_2}X_{j_2}$ and 
	$q\alpha_{j_2}\alpha_{j_2-1}\cdots\alpha_1\sim
	q\alpha_{j_1}\alpha_{j_1-1}\cdots\alpha_1$ for all $q\in
	\ds(p_iX_i\alpha_i\alpha_{i-1}\cdots\alpha_{j_2+1})$.
	In this case we can shorten the computation as in
	\cref{cl:shorten}, where 
	$v_{j_1}v_{j_1+1}\cdots v_{j_2-1}$ corresponds to the omitted
	$w_2$.
The resulting shorter computation might be possible to be repeatedly
shortened further (if it can be presented so that the conditions of~\cref{cl:shorten} are
satisfied). 
Now for each $i\geq 1$ we fix a (stair-factorized) computation
\begin{equation}\label{eq:stairseqtoi}
p_{i,0}X_{i,0}\gt{v_{i,0}}p_{i,1}X_{i,1}\alpha_{i,1}\gt{v_{i,1}}p_{i,2}X_{i,2}
		\alpha_{i,2}\alpha_{i,1}
		\ \cdots\  \gt{v_{i,n_i-1}}
		p_{i,n_i}X_{i,n_i}\alpha_{i,n_i}\alpha_{i,n_i-1}\cdots\alpha_{i,1}
\end{equation}
that has arisen by a maximal sequence of the above shortenings of the
prefix
\begin{center}
$p_0X_0\gt{v_0v_1\cdots
v_{i-1}}p_iX_i\alpha_i\alpha_{i-1}\cdots \alpha_1$
of~(\ref{eq:stairseq}). 
\end{center}
Hence
\ $p_{i,0}X_{i,0}=p_0X_0$,
\ $p_{i,n_i}X_{i,n_i}=p_iX_i$, 
\ $\alpha_{i,n_i},\alpha_{i,n_i-1},\dots,\alpha_{i,1}$ is a
subsequence of
\\
$\alpha_i, \alpha_{i-1},\dots, \alpha_1$,
\ and 
\ $p_{i,n_i}X_{i,n_i}\alpha_{i,n_i}\alpha_{i,n_i-1}\cdots\alpha_{i,1}\sim 
p_iX_i\alpha_i\alpha_{i-1}\cdots \alpha_1$.

\begin{claim}\label{cl:shortenedarelong}
	For each $\ell\in\Nat$ there is $i$ such that $n_i>\ell$
	(where $n_i$ is from~(\ref{eq:stairseqtoi})).
\end{claim}
\begin{proof}
As already discussed, 
the set of level-configurations represents infinitely many
	$\sim$-classes. The last configurations of
	computations~(\ref{eq:stairseqtoi}) represent the same
	infinite set of $\sim$-classes, and their lengths 
 thus cannot
	be bounded;
since the lengths of all $\alpha_{i,j}$ are bounded (they are
	shorter	than the longest right-hand sides of the rules in
	$R$), the claim is clear.
\end{proof}

Now we come to a crucial claim in our proof of
\cref{lem:nonregwitness}.
Besides the notation $\ds(p\alpha)$ we also 
	introduce $\es(p\alpha)$
(the by-$\varepsilon$-reached down-states of $p\alpha$), 
by putting
\begin{center}	 
	$\es(p\alpha)=\{q\mid p\alpha\gt{\varepsilon}q\}$.
\end{center}
Hence $\es(p\alpha)\subseteq\ds(p\alpha)$, and 
	$|\es(p\alpha)|\leq 1$ (due to the determinism of the DPDA
	$\mathcal{M}$).

We recall that $p\alpha\sim q\beta$ means $\calL(p\alpha)= \calL(q\beta)$.
	To handle the special case of the empty word $\varepsilon$,  
	we also define a (much) coarser equivalence $\sim_0$:
we	put $p\alpha\sim_0 q\beta$ if 
$\varepsilon$ either belongs to both
	$\calL(p\alpha)$ and $\calL(q\beta)$, or belongs to none of
	them. 

	\begin{claim}\label{cl:Ramsey}
	There is 
a constant $\textsc{B}\in\Nat$ determined by the DPDA $\mathcal{M}$ such that
for all $i\in\Nat$ where $n_i> B$ the final configuration in~(\ref{eq:stairseqtoi})
can be written as
	\begin{center}
		$p_{i,n_i}X_{i,n_i}\alpha_{i,n_i}\alpha_{i,n_i-1}\cdots\alpha_{i,1}=\bar{p}\bar{X}\beta\gamma\delta$
	\end{center}
		where the following conditions hold:
\begin{enumerate}		
	\item
		$\gamma=\alpha_{i,j}\alpha_{i,j-1}\cdots\alpha_{i,j'{+}1}$
		where 
		$n_i\geq j>j'\geq n_i{-}B$
and $p_{i,j}X_{i,j}=p_{i,j'}X_{i,j'}$
\\		(and
		$\beta=\alpha_{i,n_i}\alpha_{i,n_i-1}\cdots\alpha_{i,j{+}1}$,
	$\delta=\alpha_{i,j'}\alpha_{i,j'-1}\cdots\alpha_{i,1}$);
\item
the sets $\ds(\bar{p}\bar{X}\beta)$ and
		$\ds(\bar{p}\bar{X}\beta\gamma)$ are equal,
further being denoted by $\bar{Q}$;
		\item for each $q\in\bar{Q}$, if 
		$\es(q\gamma)=\{q'\}$, then 	
		$\es(q'\gamma)=\{q'\}$ (and $q'\in\bar{Q}$);
\item 
	each $q'\in\bar{Q}$ belongs to
				$\ds(q\gamma)$ for some
				self-containing	$q\in\bar{Q}$,		
where  $q\in\bar{Q}$ is \emph{self-containing} if $q\in
			\ds(q\gamma)$;

	\item 
there is
a state $q'\in\bar{Q}$
for which
$q'\gamma\delta\not\sim q'\delta$ and $q'\gamma\delta\sim_0
		q'\delta$.
\end{enumerate}
\end{claim}
\begin{proof}
We fix some $i$ with $n_i$ 
	larger than a constant $B$ determined by $\calM$ as described
	below
	(there are such $i$ by \cref{cl:shortenedarelong}).
For convenience 
we put $p_{i,n_i}X_{i,n_i}=\bar{p}\bar{X}$, $n_i=n$, and
	$\alpha_{i,j}=\bar{\alpha}_j$, hence the final configuration
	in~(\ref{eq:stairseqtoi}) is
$p_{i,n_i}X_{i,n_i}\alpha_{i,n_i}\alpha_{i,n_i-1}\cdots\alpha_{i,1}=
	\bar{p}\bar{X}\bar{\alpha}_{n}\bar{\alpha}_{n-1}\cdots\bar{\alpha}_{1}
	$. 
	We view the $n{+}1$ prefixes 
	\begin{center}
$\bar{p}\bar{X},\ \bar{p}\bar{X}\bar{\alpha}_{n},\
\bar{p}\bar{X}\bar{\alpha}_{n}\bar{\alpha}_{n-1},\ 
	\bar{p}\bar{X}\bar{\alpha}_{n}\bar{\alpha}_{n-1}\bar{\alpha}_{n-2},\ 
		\dots,\
		\bar{p}\bar{X}\bar{\alpha}_{n}\bar{\alpha}_{n-1}\cdots\bar{\alpha}_{1}$ 
	\end{center}		
as the vertices
of a complete graph with coloured edges. 

For 
$\bar{p}\bar{X}\bar{\alpha}_{n}\bar{\alpha}_{n-1}\cdots\bar{\alpha}_{1}=
\bar{p}\bar{X}\mu\nu\rho$, where 
	$\mu=\bar{\alpha}_n\bar{\alpha}_{n-1}\cdots\bar{\alpha}_{j{+}1}$,
	$\nu=\bar{\alpha}_j\bar{\alpha}_{j-1}\cdots\bar{\alpha}_{j'{+}1}$,
	and 
	$\rho=\bar{\alpha}_{j'}\bar{\alpha}_{j'-1}\cdots\bar{\alpha}_{1}$,
	$n\geq j>j'\geq 0$, 
the edge between 
	the vertices $\bar{p}\bar{X}\mu$ and $\bar{p}\bar{X}\mu\nu$
 has the following tuple as its \emph{colour}:
		\begin{center}
		$\left(\,p_{i,j}X_{i,j},\ p_{i,j'}X_{i,j'},\ 
		\ds(\bar{p}\bar{X}\mu),\
		\ds(\bar{p}\bar{X}\mu\nu),\ (\ds(q\nu),
		 \es(q\nu))_{q\in\ds(\bar{p}\bar{X}\mu)},\ \Qneq,\
		\Qneqzero\right)$			
		\end{center}	
where
	$\Qneq=\{q'\in\ds(\bar{p}\bar{X}\mu)\mid q'\nu\rho\not\sim
	q'\rho\}$ and 
	$\Qneqzero=\{q'\in \Qneq\mid
	q'\nu\rho\sim_0 q'\rho\}$ (and $p_{i,j}X_{i,j},\
	p_{i,j'}X_{i,j'}$ are taken from~(\ref{eq:stairseqtoi})).

	Since the set of colours is bounded (by a constant determined
	by $\calM$),
 Ramsey's theorem yields a bound $B$ guaranteeing that 
there is a
	monochromatic clique of size $3$
among the vertices
	$\bar{p}\bar{X}$,\ $\bar{p}\bar{X}\bar{\alpha}_{n}$,\ 
	$\bar{p}\bar{X}\bar{\alpha}_{n}\bar{\alpha}_{n-1}$,\ $\dots$,\ 
	$\bar{p}\bar{X}\bar{\alpha}_{n}\bar{\alpha}_{n-1}\cdots\bar{\alpha}_{n-B}$.
(We have soundly chosen $i$ so that $n=n_i$ is bigger than $B$.)
	We fix
such a monochromatic clique $\textsc{MC}$,
denoting
	its
$3$ vertices as
	\begin{center}
	$\bar{p}\bar{X}\beta$,\ $\bar{p}\bar{X}\beta\gamma$,\
		$\bar{p}\bar{X}\beta\gamma\bar{\gamma}$, 
		and its 
colour as 
		$\,\textsc{C}=(p'X',p'X',\bar{Q},\bar{Q},(\calD_q,\calE_q)_{q\in
		\bar{Q}},Q',Q'_0)$.
\end{center}
		This is sound, since the fact that
both edges 
$\{\bar{p}\bar{X}\beta, \bar{p}\bar{X}\beta\gamma\}$ and 
$\{\bar{p}\bar{X}\beta\gamma, \bar{p}\bar{X}\beta\gamma\bar{\gamma}\}$
have the same colour entails that the first component in this colour is the same as
		the second component, and the third component is the same as the
		fourth component.

We now show that the conditions $1$--$5$ 
are satisfied for the presentation of
$\bar{p}\bar{X}\bar{\alpha}_{n}\bar{\alpha}_{n-1}\cdots\bar{\alpha}_{1}$
as
$\bar{p}\bar{X}\beta\gamma\delta$,
where
$\delta=\bar{\gamma}\bar{\alpha}_{k}\bar{\alpha}_{k-1}\cdots\bar{\alpha}_1$
for the respective $k$.

Conditions $1$ and $2$ are trivial (due to the colour $\textsc{C}$).

Condition 3: Let
	$q\in\bar{Q}$ and
$\es(q\gamma)=\{q'\}$ (hence also $q'\in\bar{Q}$). Then
$\calE_{q}=\es(q\gamma)=\es(q\gamma\bar{\gamma})=\{q'\}$
(since $\textsc{MC}$ is	monochromatic).
	This entails
	$\es(q'\bar{\gamma})=\{q'\}$, hence $\calE_{q'}=\{q'\}$,
	which in turn entails
	$\es(q'\gamma)=\{q'\}$. 

Condition $4$:
We first note a general fact:
$\ds(p\mu\nu)=\bigcup_{q\in\ds(p\mu)}\ds(q\nu)$.
Since $\bar{Q}= \ds(\bar{p}\bar{X}\beta)=
\ds(\bar{p}\bar{X}\beta\gamma)=\ds(\bar{p}\bar{X}\beta\gamma\bar{\gamma})$, 
for each $q'\in\bar{Q}$ there is thus
$q\in\bar{Q}$ such that $q'\in\calD_{q}$.
We also have the following ``transitivity'':
if $q_1,q_2,q_3\in\bar{Q}$, $q_1\in\calD_{q_2}$, and
 $q_2\in\calD_{q_3}$, then $q_1\in\calD_{q_3}$ (since  $\textsc{MC}$
 is monochromatic). For any $q'\in\bar{Q}$ there is clearly a
 ``chain'' $q'=q_1, q_2, q_3, \dots, q_\ell$ where $\ell>1$,
 $q_j\in\calD_{q_{j+1}}$
 for all $j\in[1,\ell{-}1]$, and $q_j=q_\ell$ for some $j<\ell$.
By the above transitivity, $q_\ell$ is self-containing
($q_\ell\in\calD_{q_\ell}$ and thus $q_\ell\in\ds(q_\ell\gamma)$)
and 
$q'\in\calD_{q_\ell}$ (hence $q'\in\ds(q_\ell\gamma)$). 

Condition $5$:
For any three configurations at least two belong to the same
$\sim_0$-class.
Since the edges among the vertices  
$\bar{p}\bar{X}\beta$, $\bar{p}\bar{X}\beta\gamma$,
$\bar{p}\bar{X}\beta\gamma\bar{\gamma}$ 
have the same $Q'_0$ in their colour $\textsc{C}$, we get 
that $Q'_0=Q'$, and thus also 
$q'\gamma\delta\sim_0 q'\delta$ for all $q'\in \bar{Q}$ such that
$q'\gamma\delta\not\sim q'\delta$.
Now if for all $q'\in\bar{Q}$ we had $q'\gamma\delta\sim q'\delta$
(which includes the case $\bar{Q}=\emptyset$),
then we would get a contradiction with our choice
of~(\ref{eq:stairseqtoi}) since it could have been shortened as in
\cref{cl:shorten}.
\end{proof}

Now we are already close to \cref{lem:nonregwitness}:

\begin{claim}\label{cl:weakerversion} 
There are $v\in\Sigma^*$, $x,w,y,z\in\Sigma^+$, $p,q\in Q$,
$X\in\Gamma$,
$\gamma\in\Gamma^+$, $\delta\in\Gamma^*$ such that 
$p_0X_0\gt{v}pX\delta$, $pX\gt{x} pX\gamma$,	$pX\gt{w}q$,
	$q\gamma\gt{y}q$, and
	\begin{itemize}
\item
	either 
 $z\in\calL(q\delta)$ and 
	$z\not\in\calL(q\gamma^\ell\delta)$ for all $\ell>0$,
\item
			or  $z\not\in\calL(q\delta)$ and 
	$z\in\calL(q\gamma^\ell\delta)$ for all $\ell>0$.
	\end{itemize}	
\end{claim}
\begin{proof}
	We fix one $\bar{p}\bar{X}\beta\gamma\delta$ guaranteed 
	by \cref{cl:Ramsey} (satisfying the respective conditions
	$1$--$5$).  
There are  $v\in\Sigma^*$,
	$x,w,y,\bar{z}\in\Sigma^+$, $p,q\in Q$,
$X\in\Gamma$,
	$\gamma\in\Gamma^+$, $\delta\in\Gamma^*$, 
	$q'\in\ds(q\gamma)$  
	such that 
\begin{center}
	$p_0X_0\gt{v}pX\delta$,
	$pX\gt{x}pX\gamma$,
	$pX\gt{w}q$, $q\gamma\gt{y}q$, and
$\calL(q'\gamma\delta)$ and $\calL(q'\delta)$
 differ on $\bar{z}$
\end{center}
	(i.e., $\bar{z}\in
	(\calL(q'\gamma\delta)\smallsetminus\calL(q'\delta))
	\cup	(\calL(q'\delta)\smallsetminus\calL(q'\gamma\delta))$.
	\\
	(Indeed: The respective computation~(\ref{eq:stairseqtoi}) can be written
	$p_0X_0\gt{v}pX\delta\gt{x}
	pX\gamma\delta\gt{w'}
	\bar{p}\bar{X}\beta\gamma\delta$ where $x$ and $\gamma$ are
	nonempty.
	The claimed $q'$ and [nonempty] $\bar{z}$  are guaranteed by
	$5$ in
	\cref{cl:Ramsey}, and $q$ is a respective self-containing
	state from $4$.
	Since $q\in\ds(\bar{p}\bar{X}\beta)$ and $q\in\ds(q\gamma)$,
	we get
	$pX\gamma\delta\gt{w'w''}q\gamma\delta\gt{y}q\delta$, where
	$w''\neq\varepsilon$.
We also have $y\neq\varepsilon$, since
	otherwise $\ds(q\gamma)=\es(q\gamma)=\{q\}$, $q'=q$, and we
	could not have 
	$q\gamma\delta\not\sim q\delta$
and $q\gamma\delta\sim_0 q\delta$.)

	Since $q'\in\ds(q\gamma)$, we can fix $z'$ such that
	$q\gamma\gt{z'}q'$.
Hence the languages 
$\calL(q\gamma\gamma\delta)$ and $\calL(q\gamma\delta)$
differ on
	$z=z'\bar{z}$;
more generally,
	$\calL(q\gamma^{\ell+1}\gamma\delta)$ and
	$\calL(q\gamma^{\ell}\gamma\delta)$ differ on $y^\ell z$ for
	all $\ell\geq 0$.
Now we aim to find out for which $\ell$ we have
$z\in\calL(q\gamma^\ell\delta)$. 

We recall that
$\bar{Q}=\ds(\bar{p}\bar{X}\beta)=\ds(\bar{p}\bar{X}\beta\gamma)$;
hence $\bigcup_{\bar{q}\in\bar{Q}}\ds(\bar{q}\gamma)=\bar{Q}$.
Since $q\in\bar{Q}$, we get that
 $\ds(q\gamma^d)\subseteq\bar{Q}$ for all $d\in\Nat$ (by induction).
 We now
distinguish 
 two cases:
\begin{enumerate}
	\item 
For each prefix $z_1$ of $z$ and each
 $d\leq |z|$ we have:
if $q\gamma^d\gt{z_1}\bar{q}$,
		then $\es(\bar{q}\gamma)=\emptyset$.
	\item
There are a prefix $z_1$ of $z$,
		$d\leq |z|$,
		and $\bar{q},q''\in\bar{Q}$ such that 
		$q\gamma^d\gt{z_1}\bar{q}$ and
		$\es(\bar{q}\gamma)=\{q''\}$.
\end{enumerate}	
In the case $1$ we clearly have either 
$\forall \ell>|z|:z\in\calL(q\gamma^\ell\delta)$ or 
$\forall \ell>|z|:z\not\in\calL(q\gamma^\ell\delta)$ 
(here $\delta$ plays no role). 
In the case $2$ we recall that
$\bar{q}\gamma\gt{\varepsilon}q''$ entails that 
$\bar{q}\gamma^k\delta\gt{\varepsilon}q''\delta$ for all $k\geq 1$
(since $\es(q''\gamma)=\{q''\}$ by $3$ in
\cref{cl:Ramsey}).
Hence we  have
either $\forall \ell>|z|+1:z\in\calL(q\gamma^\ell\delta)$ or 
$\forall \ell>|z|+1:z\not\in\calL(q\gamma^\ell\delta)$.

Since $\calL(q\gamma^2\delta)$ and
$\calL(q\gamma^1\delta)$ differ on $z$, we deduce that there is
$\ell_0\geq 1$ such that 
		either
$z\in\calL(q\gamma^{\ell_0}\delta)$ and
$z\not\in\calL(q\gamma^{\ell}\delta)$ for all $\ell>\ell_0$,	
	or 
$z\not\in\calL(q\gamma^{\ell_0}\delta)$ and
$z\in\calL(q\gamma^{\ell}\delta)$ for all $\ell>\ell_0$.
Hence for $\bar{\delta}=\gamma^{\ell_0}\delta$ we have
		either
		$z\in\calL(q\bar{\delta})$ and
		$z\not\in\calL(q\gamma^{\ell}\bar{\delta})$ for all
		$\ell>0$,
	or 
		$z\not\in\calL(q\bar{\delta})$ and
		$z\in\calL(q\gamma^{\ell}\bar{\delta})$ for all $\ell>0$.
Since for  $\bar{v}=vx^{\ell_0}$ we have
$p_0X_0\gt{\bar{v}}pX\bar{\delta}$, the claim is proven.
\end{proof}

Claim~\ref{cl:weakerversion} is 
a weaker version of \cref{lem:nonregwitness}; it shows that
there is $L'\in\{L,\overline{L}\}$ such that 
$vx^mwy^mz\in L'$ and $vx^mwy^nz\not\in L'$ for	$m>n$.
To handle the case $m<n$, we have to find out 
for which $\ell$ we have $y^\ell z\in\calL(q\delta)$.
We thus look at the computation
	from 
	$q\delta$ on the infinite word $y^\omega$ (recalling our
	convention that this computation is infinite, stepwise reading
	the word $yyy\cdots$), and use the obvious fact that after a
	prefix this computation becomes ``periodic'' (either cycling among
	finitely many configurations, or increasing the stack
	forever).

\begin{claim}\label{cl:strongversion}	
	For any configuration $q\delta$ and words $y,z$ there are 
numbers	 $k\geq 0$ and
	$\scp > 0$ (``period'') such that 
	for all $\ell\geq k$ the remainder $(\ell\bmod \scp)$
	determines whether or not $\calL(q\delta)\ni y^\ell z$.
\end{claim}
\begin{proof}
We assume $y\neq\varepsilon$ (otherwise the claim is trivial).
For the infinite computation from $q\delta$ on $yyy\cdots$ 
	there are obviously $k_1\geq 0$,
$k_2> 0$, $\bar{q}\in Q$, and $\rho,\mu,\nu\in\Gamma^*$
such that the computation can be written
$q\delta\gt{y^{k_1}}\bar{q}\rho\nu\gt{y^{k_2}}\bar{q}\rho\mu\nu\gt{y^{k_2}}
	\bar{q}\rho\mu\mu\nu\gt{y^{k_2}}
	\bar{q}\rho\mu\mu\mu\nu\gt{y^{k_2}}\cdots$
where $\bar{q}\rho\gt{y^{k_2}}\bar{q}\rho\mu$.
	(We have $\mu=\varepsilon$ if the computation
visits only finitely many configurations, and otherwise
we consider the stair-factorization of the
	computation.)

	For each $j\in[0,k_2{-}1]$ we put
	$\bar{q}\rho\gt{y^{j}}\bar{q}\rho_j$, and 
we	
have two possible cases:
	\begin{enumerate}
		\item
There is $d_0\geq 0$ such that for all $d\geq d_0$
 performing $z$ from $\bar{q}\rho_j\mu^d\nu$ 
	does not reach $\nu$ at the bottom.
\item
There are $d_0\geq 0$, a prefix $z'$ of $z$, $q'\in Q$, and
		$\bar{d}\in[1,|Q|]$
			such that
$\bar{q}\rho_j\mu^{d_0}\gt{z'}q'$ and $q'\mu^{\bar{d}}\gt{\varepsilon}q'$.
	\end{enumerate}
	In the case $1$ either $\calL(q\delta)\ni y^{d\cdot k_2+j} z$ for all
	$d\geq d_0$, or $\calL(q\delta)\not\ni y^{d\cdot k_2+j} z$ for all
	$d\geq d_0$.
\\
In the case $2$, for each $d\geq 0$ we have   
	$q'\mu^d\gt{\varepsilon}q_{d}$ where $q_{d_1}=q_{d_2}$ if
	$d_1\equiv d_2\ (\bmod\ \bar{d})$. 
	Hence 
	for each $d\geq d_0$, the (non)membership of $y^{d\cdot
	k_2+j}z$ in $\calL(q\delta)$ is determined by $(d\bmod
	\bar{d})$.

The claim is thus clear.
\end{proof}

Now we finish the proof of \cref{lem:nonregwitness}. We take the
notation from \cref{cl:weakerversion}; 
for the respective $q\delta, y, z$ we add $k,\scp$ from
\cref{cl:strongversion}. Let $k_0$ be a multiple of $\scp$ that is
bigger than $k$.
We now view $x^{k_0}$, $y^{k_0}$, $\gamma^{k_0}$ as new
$x,y,\gamma$, respectively.
Claims~\ref{cl:weakerversion} and~\ref{cl:strongversion} now yield the
statement of \cref{lem:nonregwitness}.

\section{Conclusion and Open Problems}
\label{concl}

In this paper, we have introduced a new notion of the $\mathcal{C}$-simple problem
that reduces to each problem in $\mathcal{C}$, being thus a conceptual
counterpart to the $\mathcal{C}$-hard problem to which each problem in
$\mathcal{C}$ reduces. We have illustrated this concept on the
definition of the \nrdcfl-simple problem that reduces to each
{\nrdcfl} language under the truth-table reduction by Mealy machines.
We have proven that the {\nrdcfl} language $L_\#=\{0^n1^n\mid n\geq 1\}$ is \nrdcfl-simple, 
and thus represents the simplest languages in the class \nrdcfl. This
result finds its application in expanding the known lower bound for
$L_\#$,
namely that $L_{\#}$ cannot be recognized by the neural network model 1ANN, to all {\nrdcfl} languages. Moreover, the class DCFLS of \nrdcfl-simple problems containing the regular languages is a strict subclass of DCFL and has similar closure properties as DCFL.

We note that the hardest context-free language $L_0$ 
by Greibach~\cite{DBLP:journals/siamcomp/Greibach73},
where each $L$ in CFL is an inverse homomorphic image of $L_0$ or
$L_0\smallsetminus\{\varepsilon\}$, 
can be viewed as CFL-hard w.r.t. a~many-one reduction based on Mealy
machines realizing the respective homomorphisms.
Our aims in the definition of
\nrdcfl-simple problems cannot be achieved by 
such a many-one reduction, hence we have generalized it to a truth-table reduction.
We can alternatively consider a general Turing reduction
that is implemented by a Mealy machine which queries the oracle at
special query states, each associated with a corresponding query
suffix, while its next transition from the query state depends on the
given oracle answer. The oracle Mealy machine then accepts an input
word if it reaches an accept state after reading the input. 
The language $L_\#$ proves to be
\nrdcfl-simple under this Turing reduction allowing for an
unbounded number of online oracle queries; this can be shown by \cref{cl:weakerversion} (a weaker version of \cref{lem:nonregwitness}).

It is natural to try 
extending our result to non-regular
nondeterministic (or at least unambiguous) context-free languages,
by possibly showing that $L_\#$ is \nrcfl-simple. 
Another important challenge for further research is looking for $\mathcal{C}$-simple problems for other complexity classes $\mathcal{C}$ and suitable reductions. This could provide an effective tool for strengthening lower-bounds results known for single problems to the whole classes of problems, which deserves a deeper study.

\subsection*{Acknowledgements}

Presented research has been partially supported by the Czech
Science Foundation, grant GA19-05704S, and 
by the institutional support RVO: 67985807 (J.\ \v{S}\'{\i}ma).
J.\ \v{S}\'{\i}ma also thanks Martin Pl\'atek 
for his intensive collaboration at the first stages of this research.

\bibliography{dcflsimpl}
\end{document}